\documentclass[11pt]{article}
\pdfoutput=1
\usepackage[latin1]{inputenc}
\usepackage{amsmath,amsthm,amssymb}
\usepackage[pdftex]{graphicx}
\newcommand\B{{\bf B}}

\newcommand\e{{\bf e}}
\newcommand\w{{\bf w}}
\newcommand\un{{\bf u}}
\newcommand\n{{\bf n}}

\newcommand\x{{\bf x}}
\newcommand\y{{\bf y}}
\newcommand\p{{\bf p}}
\newcommand\q{{\bf q}}

\newcommand\Cn{{\bf C}}

\newcommand\m{{\bf m}}

\title{Inertial forces in the Navier-Stokes equation}
\author{Philipp Lohrmann  \thanks{Supported in part by the European Research Counsil under FP7 \textquotedblleft New connections between dynamical systems and Hamiltonian PDE with small divisor phenomena"}}

\theoremstyle{plain}
\newtheorem{pr}{Proposition}[section]
\newtheorem{Lm}[pr]{Lemma}
\newtheorem{thm}[pr]{Theorem}

\newtheorem{cor}[pr]{Corollary}

\theoremstyle{definition}

\theoremstyle{remark}
\newtheorem{rem}{Remark}[section]

\numberwithin{equation}{section}

\begin{document}

\maketitle

\begin{abstract}
We estimate the inertial forces in the 3d Euler/Navier-Stokes equation in function of $|\un|_\infty$ and $|\nabla \un |_\infty$ and provide an application to the well-posed problem.  
\end{abstract}

\section{Introduction}

{\it Background.} The purpose of this paper is to present a new method to analyze the inertial force field in the the 3d Euler/Navier-Stokes equation. A better understanding of the inertial force field is desirable in many problems in fluid mechanics, and essential for further progress in the well-posed problem for the Navier-Stokes equation.
Nowadays the question whether a solution can develop a finite time singularity from smooth initial data remains open, although 
many heuristic facts indicate that the solutions do not blow up. 
T. Hou and Z. Li, in \cite{HuLi}, emphasized that the standard approaches (see for instance \cite{Kato}) to the Cauchy problem don't take into account the whole structure of the non-linearity: the efforts are focused on the use of the diffusion term in order to control the vorticity stretching term, considered as the main source of difficulties, whereas the convection term is \textquotedblleft neglected". The reason for this is that energy estimates are used essentially, and the convection term does not contribute to the energy norm of the velocity field, or the $L^p$-norms of the vorticity field. Hu and Li then dropped the convection term from the Navier Stokes equation and proved that counterparts of many results for the Navier-Stokes equation hold in this toy model (on axial symmetric domains). On the other hand, they provided numerical evidence that in this toy model, finite times singularities may develop from smooth initial data. Their conclusion was that in the full Navier-Stokes equation, the convection term plays an essential stabilizing role. 
A future use of the whole structure of the nonlinearity requires a better understanding of the inertial force field $\nabla P$. 

\vspace{0.2cm}

{\it Results.} In this paper we consider the 3d Euler/Navier-Stokes equation
\begin{equation} \label{ns}
\left\{\begin{array}{l} \partial_t \un + \un.\nabla \un + \nabla P = \nu \Delta \un, \qquad \nu=0,1 \\  
                                                    \text{div } \un =0 \end{array} \right.
\end{equation}
without external forces. As our purpose is to present a new method to analyze the inertial force field, we work in the simplest setting possible, and our assuptions are not the weakest possible.

\begin{thm} \label{thm2}
There exists a constant $\beta >0$ so that for any solution $(\un,P) $ of (\ref{ns}) one has at any time at which $\un$ and $\nabla \un$ belong to $L^\infty ( \mathbb R
^3) \cap C^\infty( \mathbb R
^3)$ 
\begin{equation} \label{relthm2}
|\nabla P |_\infty < \beta |\nabla \un|_\infty |\un |_\infty.
\end{equation} 
\end{thm}
We do not prove Theorem \ref{thm2} by the usual harmonic analysis methods to estimate singular integrals. Instead we give a {\it local} argument, i.e., we use an analogy with electrostatics, and study the spacial distribution 
of the \textquotedblleft charges" generating the inertial force field. To this end we provide some surprising direct estimates, relying  essentially on the divergence-free condition.

\begin{rem}
As we use local arguments, couterparts of Theorem \ref{thm2} hold for other underlying spaces, for example $\mathbb T^3$.
Our method can also be used for various other problems, for example,
estimating $\nabla P$ in space regions where the velocity resp velocity gradient is small with respect $|\un|_\infty$ resp $|\nabla \un|_\infty$.
\end{rem}

We also want to present some application. Here we only consider the case of the Navier-Stokes equation, i.e., we assume $\nu=1$, and consider initial data $\un^0$ in $L^3(\mathbb R^3) \cap L^\infty(\mathbb R^3)$. Then (see \cite{ESS}) $\un^0$ generates a unique solution $\un$ in $L^{3,\infty}(\mathbb R^3 \times [0,t_b])$, smooth in space-time for any $\x \in \mathbb R^3$ and
any $0< t < t_b(\un^0)$, where $t_b \equiv t_b(\un^0)$ denotes its (eventual) blow-up time. If she does not blow up, then we set $t_b(\un^0):= +\infty$.

\begin{pr} \label{thm3}
There exists a constant $\beta_2 >0$ so that for any initial data $\un^0$ as above, there exists $0\leq t_0\equiv t_0(\un^0)< t_b(\un_0)$, so that for any $t_0 \leq t < t_b(\un^0)$  
\begin{equation*}
|\nabla \un(\cdot,t)|_\infty < \beta_2 \text{max}_{0\leq s\leq t}|\un(\cdot,s)|_\infty^2. 
\end{equation*}
Moreover, for any $0 \leq t \leq t_0$ one has
$|\un(\cdot, t)|_\infty \leq \frac{9}{8}|\un^0|_\infty$.
\end{pr}

Finally, Propositon \ref{thm3} admits as an immediate corollary a simple {\it nonlinear} condition on the initial data, guaranteeing existence of
a classical solution for all time. As we use energy conservation here, we again do not use the whole structure of the nonlinearity. Moreover, any initial data satisfying our condition is small in $L^3$, thus this corollary is weaker than several results obtained with the classical approaches.
\begin{cor}\label{intro} 
There exists a constant $\alpha >0$ so that any initial data  $\un^0$ as above, with in addition $|\un^0|_2 < \infty$, and with
\begin{equation}  \label{condition}
|\un^0 |_\infty  < \alpha(|\un^0 |_2)^{-2}
\end{equation}
generates a global in time smooth solution of (\ref{ns}) (i.e. in $C^\infty (\mathbb R^3 \times (0, \infty)$). 
\end{cor}
To prove Corollary \ref{intro},
we use the estimate of the inertial forces to show that under the assumptions we made, the $L^\infty$-norm of the velocity field remains bounded as long as it remains smooth in the space variable. The global in time smoothness is then from the fact that when the time approaches the (eventual) blow-up
time ($\equiv$ smallest time for which the solution is not smooth in the space variable), then the $L^\infty$-norm of the solution tends to infinity.


\vspace{0.2cm}

{\it Notations.}
Small greek letters denote constants, i.e. real numbers which doesn't depend on the initial data $\un^0$, or on other variable quantities. Elements of $\mathbb R^3$  are denoted by small bold latin caracters. 
For $\x=(x_1,x_2,x_3) \in \mathbb R^3$, we denote by $|\x|:=\sqrt{x_1^2+x_2^2+x_3^2}$ the Euclidean norm of $\x$.
We set $|\un|_\infty (:=\|\un\|_{L^ \infty}):= \text{max}_{i \in \{1,2,3\}}\sup_{\x \in \mathbb R^3} |u_i(\x)|$, $\un=(u_1,u_2,u_3)$, and
$|\nabla \un|_\infty(:=\|\nabla \un\|_{L^\infty}):= \max_{i \in \{1,2,3\}} |\nabla u_i|_\infty$. 
We denote by $d\x$ either the standard volume form on $\mathbb R^3$, or  the standard area form on a two dimensional subvariety of $\mathbb R^3$.

\section{Proof of Theorem \ref{thm2} } \label{sect3}

\textit{Poisson equation}. We begin by recalling some elementary facts about similarities between the Navier-Stokes equation and some problems in Electrostatics.
As usual we take the pressure $P$ in the Euler/Navier-Stokes system as the solution of the Poisson equation 
\begin{equation} \label{chaden}
\Delta P=-\nabla.(\un. \nabla \un)
\end{equation} 
vanishing at infinity. In a similar way than in Electrostatics, we consider the right hand side of (\ref{chaden}) as a density of charges. For each bounded domain $B \subseteq \mathbb R^3$, having a (piecewise) differentiable two dimensional compact manifold as boundary, the charge $C(B)$ is defined by
\begin{equation} \label{defcharge}
C(B)=\int_{\partial B} (\nabla P) .\n \,d\x =-\int_{\partial B} (\un. \nabla \un).\n \,d\x,
\end{equation}
where $\n$ denotes the unitary vector field on $\partial B$, orthogonal to $\partial B$, and pointing outwards $B$. 
Coulomb's law allows to recover for each time $t$ the inertial forces $\nabla P$ 
from the velocity field at time $t$ by
\begin{equation} \label{PoisP}
\nabla P(\x)=\int_{\mathbb R^3} \frac{-\nabla.(\un. \nabla \un)(\y)(\x-\y)}{|\x-\y|^3}\, d\y 
\end{equation} 
Our proof of Theorem \ref{thm2} relies on some estimate concerning the distribution of the charges (\ref{defcharge}) in space, where we use in essentially that $\un$ is divergence-free. 
The key lemma is Lemma \ref{lemma1} below. In the sequel we always designate by  $\n$ a smooth unitary vector field, orthogonal to some differentiable submanifold of $\mathbb R^3$, depending on context. For any subset $A \subseteq \mathbb R^3$ we set
$w(A):= \sup _{\x,\x' \in A }|\un(\x)-\un(\x')|$.


\begin{Lm} \label{lemma1}
There exists a constant $\lambda >0$ so that for any  $C^1\cap L^\infty$ divergence-free vector field  $\un: \mathbb R^3 \rightarrow \mathbb R^3 $, and
any rectangle
$R \subseteq \mathbb R^3$ (contained in a 2d affine subspace of $\mathbb R^3$) with side lengths $l_1,l_2>0$, one has
\begin{equation*} 
\left|\int_R (\un. \nabla \un).\n \,d\x   \right| \leq \lambda \text{max}(l_1,l_2) w^2(R).
\end{equation*} 
\end{Lm}

\begin{proof}

Let $\un=(u_1,u_2,u_3),  l_1, l_2$ and $R$ be as in Lemma \ref{lemma1}. Without loss of generality we may assume 
\begin{equation*} 
R=\{\x=(x_1,x_2,x_3)\,|\, 0 \leq x_1 \leq l_1, 0 \leq x_2 \leq l_2, x_3=0  \}.
\end{equation*}
First recall that
\[
(\un. \nabla \un).\n= u_3 \partial_3 \un_3+u_2 \partial_2 u_3 + u_1 \partial_1 u_3.
\]
Integrating by parts one gets
\begin{equation*}
\int_R u_2 \partial_2 u_3\,d \x= \left.\int_0^{l_1} (u_2u_3)(x_1,\cdot,0)\right|_0^{l_2}\,dx_1- \int_0^{l_1}\int_0^{l_2}u_3 \partial_2 u_2\, dx_2dx_1,
\end{equation*}
and similarly
\begin{equation*}
\int_R u_1 \partial_1 u_3\,d \x= \left.\int_0^{l_2} (u_1u_3)(\cdot,x_2,0)\right|_0^{l_1}\,dx_2- \int_0^{l_2}\int_0^{l_1}u_3 \partial_1 u_1 \,dx_1dx_2.
\end{equation*}
As by the divergence free condition one has
\begin{equation*}
\int_R u_3 \partial_1 u_1 \,d\x +\int_R u_3 \partial_2 u_2 \,d\x =-\int_R u_3 \partial_3 u_3 \,d\x
\end{equation*}
it follows that
\begin{equation*}
\int_R (\un. \nabla \un).\n \,d\x =2 \int_R u_3 \partial_3 u_3 +\left.\int_0^{l_1} (u_2u_3)(x_1,\cdot,0)\right|_0^{l_2}\,dx_1 +\left.\int_0^{l_2} (u_1u_3)(\cdot,x_2,0)\right|_0^{l_1}\,dx_2.
\end{equation*}
Then, as one has  $\left|\left.\int_0^{l_1} (u_2u_3)(x_1,\cdot,0)\right|_0^{l_2}\,dx_1\right| \leq 2l_1w^2(R) $ resp $\left|\left.\int_0^{l_2} (u_1u_3)(\cdot,x_2,0)\right|_0^{l_1}\,dx_2 \right| \leq l_2w^2(R) $, in order to prove Lemma \ref{lemma1}, it remains to show that there exists a constant $\tilde \lambda >0$ so that 
\begin{equation*}
\left|\int_R u_3 \partial_3 u_3\, d\x \right|\leq \tilde \lambda \text{min}(l_1,l_2) w(R).
\end{equation*}
To this end consider the vector field $\tilde \un:=(\tilde u_1,\tilde u_2, \tilde u_3)$, defined by $\tilde \un(\x):= \un(\x)$ if $u_3(\x)\geq 0$, and $\tilde \un(\x):= -\un(\x)$ if $u_3(\x)< 0$. 
Note that $\tilde \un$ is smooth and divergence free on some subset of full measure of $\mathbb R^3$, and
\begin{equation*}
\int_R \tilde u_3 \partial_3 \tilde u_3\,d\x=\int_R u_3 \partial_3 u_3\,d\x.
\end{equation*}
Then, $\int_R \tilde u_3 \partial_3 \tilde u_3\,d\x=-\int_R \tilde u_3 (\partial_2 \tilde u_2+\partial_1 \tilde u_1)\,d\x$, and as $\tilde u_3$ is of constant sign, in order to establish Lemma \ref{lemma1}, it suffice to show 
that there are constants $ \lambda_1, \lambda_2 >0$ so that
\begin{equation*}
\left|\int_R \partial_{1,2} \tilde u_{1,2} \,d\x   \right| \leq  \lambda_{1,2} \text{max}(l_1,l_2) w(R).
\end{equation*}
For any $0 \leq x \leq l_1$, let $\ell_x$ be the 
line segment $\{(x_1,x_2,x_3)\,|\, x_1=x, 0\leq x_2 \leq l_2, x_3=0\}$. Then for any $0 \leq x \leq l_1$ we have $\left| \int_{\ell_x} \partial_2 \tilde u_2 d \,x_2 \right| \leq w(R)$, thus $\left|\int_R \partial_2 \tilde u_2 \,d\x  \right| \leq  l_1 w(R)$. Similarily one obtains
$\left|\int_R \partial_1 \tilde u_1 \,d\x   \right| \leq  l_2 w(R)$.  
\end{proof}

{\it Sketch of the proof of Theorem \ref{thm2}.} Recall that we denote $\x=(x_1,x_2,x_3)$ and $\nabla P=(\partial_1 P,\partial_2 P, \partial_3 P)$. Without loss of generality it suffice to estimate $\partial_1 P(0)$ i.e. the $x_1$-component of the inertial forces at $\x=0$. To this end we decompose the half-space $\mathbb R^3_+:=\{\x \in \mathbb R^3\,|\, x_1>0\}$ resp  $\mathbb R^3_-:=\{\x \in \mathbb R^3\,|\, x_1<0\}$ into some disjoint countable union of building blocks, 
in such a way that for each point $\x$ in a given block, the angle between $\x$ and the $x_1$-axis is approximatively the same, as well as the distance  between $\x$ and the origine. We estimate the total charge inside the blocks using some variant of Lemma \ref{lemma1} and (\ref{defcharge}). Moreover, using again a variant of Lemma \ref{lemma1}, we show that the charges inside a block don't give rise to a strong dipol. The contribution to $\partial_1 P(0)$ of each block is then essentially given by its position and total charge.
 
{\it Block decomposition.} For any $n \in \mathbb Z$ introduce 
the cylinder $\Cn_n :=\{\x \in \mathbb R^3\,|\, 2^n<x_1<2^{n+1}, x_2^2+x_3^2 <4^{n+1} \}$
and $\B_n := \{\x \in \mathbb R^3\,|\, x_1<2^{n}, 4^n<x_2^2+x_3^2 <4^{n+1} \}  $.
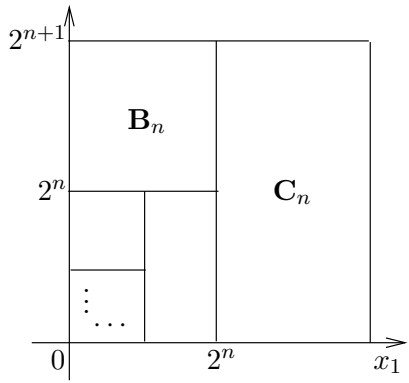
\vspace{4.7cm}

\begin{picture}(20,20)(-130,0)
\put(67,5){$2^n$}  \put(130,5){$x_1$}   \put(92,69){$\Cn_n$}    \put(8,5){$0$} \put(3,69){$2^n$}  \put(-8,126){$2^{n+1}$}
\put(37,96){$\B_n$}    \put(24,19){$\cdots$}  \put(20,26){$\vdots$}
\includegraphics{document1.pdf}
\end{picture}

\begin{Lm} \label{cor1lemma1} 
There exists a constant $\lambda >0$ so that for any divergence-free vector field $\un:\mathbb R^3 \rightarrow \mathbb R^3$ in $C^1 \cap L^\infty$ 
\begin{itemize}
\item[(i)] for any $c \in \mathbb R$, and any $r,r_1,r_2 >0$,  one has for 
the disc $D=\{\x \in \mathbb R^3\,|\, x_1=c, x_2^2+x_3^2 < r^2 \}$ resp annulus $A=\{\x \in \mathbb R^3\,|\, x_1=c, r_1^2 <x_2^2+x_3^2 < r_2^2 \}$ the estimate
\begin{equation*}
\left|\int_D (\un. \nabla \un).\n \, d\x   \right| \leq \lambda r w^2(D)\quad \text{ resp }\quad \left|\int_A (\un. \nabla \un).\n \, d\x   \right| \leq \lambda r_2 w^2(A);
\end{equation*}
\item[(ii)] for any $c_1,c_2 \in \mathbb R$, and any $r>0$, one has for the cylinder
$C=\{\x \in \mathbb R^3\,|\, c_1<x_1<c_2, x_2^2+x_3^2 = r^2 \}$ the estimate 
\begin{equation*}
\left|\int_C (\un. \nabla \un).\n \,d\x   \right| \leq \lambda \text{max}(c_2-c_1,r) w^2(C).
\end{equation*}
\end{itemize}
\end{Lm}

\begin{proof}
To prove item (i), consider the rectangle $\hat D:=\{x_1=c, |x_2|< r, |x_3|<r\}$ and the vector field $\hat \un: \hat D \rightarrow \mathbb R^3$, with $\hat \un(\x):=\un(x)$ for $\x \in D$ and $\hat \un(x):=0$ otherwise. The first statement in item (i) then follows by arguing for $\hat D$ and $\hat \un(x)=0$ in the same way as in the proof of Lemma \ref{lemma1} for $R$ and $\un$.
To prove the statement concerning the annulus $A$, cut it into two equal pieces and argue, separately for each piece, in the same way. 
To prove item (ii), cut the cylinder $C$ along a straight line parallel to the $x_1$-axis, and consider it as a rectangle ruled on $\{x_2^2+x_3^2 = r^2\}$. Then unroll $C$ and apply Lemma \ref{lemma1}.
\end{proof}

\begin{cor}[charge estimate] \label{estcharge}
There exists a constant $\lambda >0$ so that for any divergence-free vector field $\un:\mathbb R^3 \rightarrow \mathbb R^3$ in $C^1 \cap L^\infty$, for any $n \in \mathbb Z$, the charges $C(\B_n)$ and $C(\Cn_n)$ inside $\B_n$ resp $\Cn_n$, defined by (\ref{defcharge}), satisfy the estimates
\[
|C(\B_n)| \leq \lambda w^2(\B_n) 2^{n} \qquad \text{ and }\qquad |C(\Cn_n)| \leq \lambda w^2(\Cn_n) 2^{n}.
\]
\end{cor}

\begin{proof}
Let $n \in \mathbb Z$, $\x \in \B_n$, and take a coordinate system in uniform translation with a velocity $\un(\x)$. We obtain the 
claimed bound of $C(\B_n)$ by estimating the flow of $\un.\nabla \un$ through each face of $\B_n$, in the new coordinates, by Lemma \ref{cor1lemma1}. The claimed estimate of $C(\Cn_n)$ is obtained arguing in the same way.
\end{proof}

To be able to use Corollary \ref{estcharge} in order to bound the inertial forces at $\x=0$, we need to ensure that the charges
inside the blocks  $\B_n$ and $\Cn_n$ do not constitute strong dipoles. To this end consider $h:\mathbb R^3_+ \rightarrow \mathbb R$,
$h(\x):= \frac{\x.\e_{x_1}}{|\x|^3}$, $\e_{x_1}:=(1,0,0)$, the $x_1$-component of the force field generated by a point unitary charge at $\x \in \mathbb R^3_+$. For any $x \in h(\mathbb R^3_+)=(0,+\infty)$, the level set  $L_x=\{h=x\}$ admits a rotational symmetry with respect to the 
$x_1$-axis, and one has the following similarity. For any $s>0$
\begin{equation} \label{similarity}
\x \in L_{x} \Rightarrow s^{-2}\x \in L_{sx}.
\end{equation}

By an implicit function theorem argument one sees that the level set  $L_x=\{h=x\}$ is a
smooth surface for any $x \in (0, +\infty)$. Moreover, one shows that whenever $L_x \cap \Cn_n \not= \emptyset$ resp $L_x \cap \B_n \not= \emptyset$, then the intersection $L_x \cap \partial \Cn_n$ resp $L_x \cap \partial \B_n$ is a union of one or two circles, having their centers on the $x_1$-axis. 
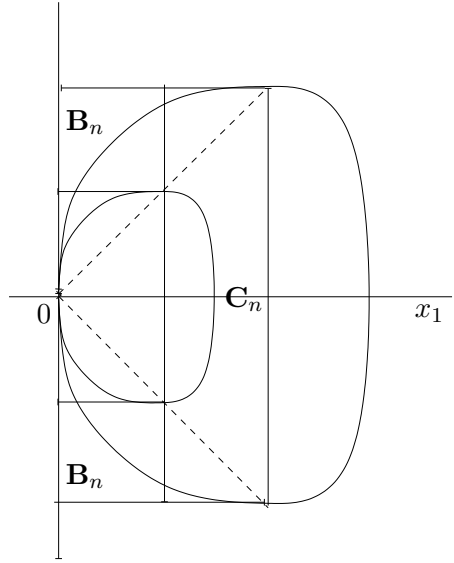
\vspace{7cm}

\begin{picture}(20,20)(-130,0)
 \put(154,92){$x_1$} \put(11,90){$0$} \put(82,96){$\Cn_n$} \put(22,163){$\B_n$} \put(22,29){$\B_n$}
\includegraphics[scale=0.7]{document3.pdf}
\end{picture}

Note that for $\x \in L_x \cap \{x_1^2=x_2^2+x_3^2\}$, the level sets $L_x$ have a tangent parallel to the $x_1$-axis.

\begin{Lm} \label{cor2lemma1}
There exists a constant $\lambda >0$ so that for any divergence-free vector field $\un:\mathbb R^3 \rightarrow \mathbb R^3$ in $C^1 \cap L^\infty$, for any
$n \in \mathbb Z$, and any $x \in h(\B_n)$ resp $x \in h(\Cn_n)$ 
\begin{equation}  \label{estimate27}
\left|\int_{L_x \cap \B_n} (\un. \nabla \un).\n \, d\x   \right| \leq \lambda 2^n w^2(L_x \cap \B_n) \quad \text{ and }\quad  \left|\int_{L_x \cap \Cn_n} (\un. \nabla \un).\n \, d\x   \right| \leq \lambda 2^n w^2(L_x \cap \Cn_n).
\end{equation}
\end{Lm}

\begin{proof}
Let $\m:\mathbb R^3_+\rightarrow \mathbb R^3$ be the vector field orthogonal to the level sets $L_x$, $x \in (0,+\infty)$, with $|\m(\x)|=1$ $\forall \x \in \mathbb R^3_+$, and pointing into the direction of increasing values of $x$. 
Next fix $x_0 \in (0, +\infty)$, and 
subdivise $L_{x_0} $ into three zones $Z_1^{x_0},Z_2^{x_0},Z_3^{x_0}$ with
$Z_1^{x_0}:=\{\x \in L_{x_0} \,|\, \pi/3 <| \left\langle \x,\e_{x_1}\right\rangle|< \pi/2 \}$ , $Z_2^{x_0}:=\{\x \in L_{x_0} \,|\, \pi/6<|\left\langle \x,\e_{x_1}\right\rangle|<\pi/3\}$  
and $Z_3^{x_0}:=\{\x \in L_{x_0} \,|\, |\left\langle \x,\e_{x_1} \right\rangle|<\pi/6\}$, where $\left\langle .,.\right\rangle$ denotes the angle, and $\e_{x_1}:=(1,0,0)$. First we consider the case $i=1,3$. By an implicit function argument one shows that the projection $p^{(i,x_0)}:Z_i^{x_0}\rightarrow \mathbb R^2$, $\x=(z_1,x_2,x_3)\mapsto (x_2,x_3)$, $i=1,3$, is an immersion, i.e. $p^{(i,x_0)}(Z_i^{x_0})$ are coordinates on $Z_i^{x_0}$. In order to extend these coordinates on some neighborhood $W_i^{x_0}$ of $Z_i^{x_0}$ in $\mathbb R^3_+$, we complete $(x_1,x_2)$ by a third coordinate $y$, in such a way that $Z_i^{x_0}=\{y=0\}$  and $\partial_y (x_2,x_3,y)=\m$, $i=1,3$. 
Then the divergence-free condition in (\ref{ns}) for $\un=(u_{x_2}, u_{x_3}, u_y)$ becomes at any $\x \in Z_i^{x_0}$ 
\begin{equation} \label{jacdiv}
\partial_{x_2}(g^{(i,x_0)}u_{x_2})+\partial_{x_3}(g^{(i,x_0)}u_{x_3})+\partial_{y}(g^{(i,x_0)}u_y)=0, 
\end{equation}
where $g^{(i,x_0)}$ is the Jacobian determinant of the projection $p^{(i,x_0)}$, $i=1,3$. Moreover, by an implicit function argument and with (\ref{similarity}), one shows the existence of 
a constant $\alpha >0$, not depending on $x_0$, so that for any $\x \in Z_i^{x_0}$, $i=1,3$, 
\begin{equation} \label{boundjac}
  \alpha <|g^{(i,x_0)}(\x)| \leq  1.
\end{equation} 
Now let $n \in \mathbb Z$ and assume that $Z_1^{x_0} \cap \Cn_n \not=0$. Note that $p^{(i,x_0)}(Z_i^{x_0})$ is an annulus in $\mathbb R^2$ for $i=1$, and a disc for $i=3$. Arguing as in the proofs of Lemma \ref{lemma1} and Lemma \ref{cor1lemma1} (i), and taking into account (\ref{jacdiv}) and (\ref{boundjac}), one shows that there is a constant $\lambda_1 >0$ so that
\begin{equation} \label{estimatea}
 \left|\int_{Z_i^{x_0} \cap \Cn_n} (\un. \nabla \un).\m \,d\x   \right| \leq \lambda_1 2^n w^2(Z_i^{x_0} \cap \Cn_n), \quad i=1,3. 
\end{equation}
In the case $i=2$, the estimate
\begin{equation} \label{estimateb}
\left|\int_{Z_2^{x_0} \cap \Cn_n} (\un. \nabla \un).\m \,d\x   \right| \leq \lambda_2 2^n w^2(Z_2^{x_0} \cap \Cn_n)
\end{equation}
is shown arguing in the same way with the following difference: instead of projecting $Z_2^{x_0}\cap \Cn_n$ on the hyperplane $\{x_1=0\}$, one projects on the cylinder $C^n:=\{x_2^2+x_3^2= 4^n  \}$. At the end of the proof, instead of arguing as in the proof of item (i), one argues as in the proof of item (ii) of Lemma \ref{cor1lemma1}. 
Putting (\ref{estimatea}) and (\ref{estimateb}) together one has the second estimate in (\ref{estimate27}). Finally one proves the first estimate arguing in the same way.
\end{proof}

Next we introduce for any subset $A \subseteq \mathbb R$ the contribution $\partial_1 P_A(0)$ to  $\partial_1 P(0)$ of the charges located in $A$ by 
\begin{equation*}
\partial_1P_A(0):= \int_{A} \frac{-\nabla.(\un. \nabla \un)(\y)\y.\e_{x_1}}{|\y|^3}\, d\y.
\end{equation*}  

\begin{Lm} \label{lmclose}
There exists  a constant $\mu >0$ so that for any divergence-free vector field $\un:\mathbb R^3 \rightarrow \mathbb R^3$ in $C^1 \cap L^\infty$, for any $n \in \mathbb Z$ 
\begin{equation} \label{estlmclose}
|\partial_1P_{\B_n}(0)|\leq \mu 2^{-n}w^2(\B_n) \qquad \text{ and } \qquad |\partial_1P_{\Cn_n}(0)|\leq \mu 2^{-n}w^2(\Cn_n).  
\end{equation} 
\end{Lm}

\begin{proof}
(i) Let $n \in \mathbb Z$ and set $h(\Cn_n):=(c_1^{(n)},c_2^{(n)})$. For any $x \in (c^{(n)}_1,c^{(n)}_2)$ set $L^{(n)}_x:=L_x \cap \Cn_n$. Note that there exists two constants $\gamma_1, \gamma_2 >0$ so that $\frac{\gamma_1}{4^{n}}\leq c_1^{(n)} < c_2^{(n)} \leq \frac{\gamma_2}{4^n} $. By Lemma \ref{cor2lemma1} there exists a constant $\lambda_1>0$ so that for any $x \in (c_1^{(n)},c_2^{(n)})$ one has
\begin{equation} \label{finaout}
\left|\int_{L_x^{(n)}} (\un. \nabla \un).\n \,d\x \right| \leq  \lambda_1 2^{n} w^{2}(\Cn_n).
\end{equation}
Next, for any $x \in (c_1^{(n)},c_2^{(n)})$ set $f(x):=\partial_x F(x)$ with $ \displaystyle F(x):= C\left(\cup_{c_1^{(n)} \leq y \leq x}L_y^{(n)}\right)$ and note that 
\begin{equation*} 
\partial_1P_{\Cn_n}(0)= \int_{c_1^{(n)}}^{c_2^{(n)}}x f(x) \,dx.  
\end{equation*} 
Then, integrating by parts, one obtains
\begin{equation} \label{partintegr}
\partial_1P_{\Cn_n}(0)= \int_{c_1^{(n)}}^{c_2^{(n)}}x f(x) \,dx= \left[xF\right]_{c_1^{(n)}}^{c_2^{(n)}}-\int_{c_1^{(n)}}^{c_2^{(n)}}F(x) \,dx.
\end{equation}
To estmate $F(x)$, i.e. the charge inside $\cup_{c_1^{(n)} \leq y \leq x}L_y^{(n)} \subseteq \Cn_n$, note that the boundary
$\partial \left(\cup_{c_1^{(n)} \leq y \leq x}L_y^{(n)} \right)$ is the union of $L_x^{(n)}$, one or two annuli 
inside the disc $\{ x_1=2^n, x^2_2+x_3^2\leq 4^{n+1}\}$ resp $\{ x_1=2^{n+1}, x^2_2+x_3^2 \leq 4^{n+1} \}$, and a cylinder inside $\{x^2_2+x_3^2=4^{n+1} \}$. 
Then by (\ref{finaout}), Lemma \ref{cor1lemma1} and (\ref{defcharge}) there is a constant $\lambda_2 >0$ so that for any intervall $(c_1^{(n)},x]\subseteq (c_1^{(n)},c_2^{(n)})$
\begin{equation*} 
\left| F(x) \right| \leq \lambda_2 w^2(\Cn_n) 2^{n}.  
\end{equation*} 
Together with (\ref{partintegr}) one finally gets
\begin{equation*}
|\partial_1P_{\Cn_n}(0)|\leq 2(c_2^{(n)}-c_1^{(n)})\lambda_2 w^2(\Cn_n) 2^{n} \leq 2 \frac{\gamma_2-\gamma_1}{4^{n}}w^2(\Cn_n) 2^{n}:=\mu2^{-n}w^2(\Cn_n).
\end{equation*}
The first estimate in (\ref{estlmclose}) is proved arguing in the same way.

\end{proof}

\begin{proof}[Proof of Theorem \ref{thm2} ]

In view of the definition of the norms one has $|\un(\x)-\un(x')| \leq 3|\nabla \un|_\infty|\x-\x'|$.
Let $n_0 \in \mathbb Z$ be the integer so that  
\begin{equation} \label{defn0}
2^{n_0-1} 3|\nabla \un|_\infty  <  |\un|_\infty \leq 2^{n_0}3|\nabla \un|_\infty. 
\end{equation}
With $\sup_{\x,\x' \in \B_n}|\x-\x'| \leq 2^{n+3}$
and $\sup_{\x,\x' \in \Cn_n}|\x-\x'| \leq 2^{n+3}$ one then has
\begin{equation*}
w(\B_n)<  3 |\nabla \un |_\infty  2^{n+3},  \text{ if } n <n_0 \quad \text{and}
        \quad       w(\B_n)<   | \un |_\infty, \text{ if } n \geq n_0.
\end{equation*}
With Lemma \ref{lmclose} it follows that there exists $\gamma >0$
so that
\begin{equation*}
|\partial_1P_{\B_n}(0)|\leq  \gamma 2^n |\nabla \un|_\infty^2,\,\,n < n_0, \qquad \text{and} \qquad |\partial_1P_{\B_n}(0)|\leq \gamma 2^{-n} |\un |_\infty^2,\,\, n \geq n_0. 
\end{equation*} 
By the left hand side inequality in (\ref{defn0}) one has $\displaystyle \sum_{n <n_0}|\partial_1P_{\B_n}(0)|\leq \frac{2}{3}\gamma |\un|_\infty |\nabla \un|_\infty $, 
whereas the right hand side inequality in (\ref{defn0}) leads to
\begin{equation*} 
\sum_{n \geq n_0}|\partial_1P_{\B_n}(0)| \leq \sum_{n \geq 0} 3\gamma 2^{-n} |\un|_\infty |\nabla \un|_{\infty} \leq 6 \gamma |\un|_\infty |\nabla \un|_{\infty}.
\end{equation*}
Arguing in the same way one obtains similar estimates for $\partial_1P_{\Cn_n}(0)$, $n \in \mathbb Z$. In total 
\begin{eqnarray*}
|\Delta P.\e| &\leq& \sum_{n <n_0}|\partial_1P_{\B_n}(0)|+ |\partial_1P_{\Cn_n}(0)|+\sum_{n \geq n_0}|\partial_1P_{\B_n}(0)|+|\partial_1P_{\Cn_n}(0)| \\
               & & +(\text{same contribution from }\mathbb R^3_-) \\
               &\leq&  \frac{80}{3} \gamma  |\un|_\infty |\nabla \un|_{\infty}.
\end{eqnarray*}
\end{proof}

\section{Friction}  \label{friction}

\begin{Lm} \label{lmheat}
There exists a constant $\alpha >0$ so that for any $L^\infty$-function
$f: \Bbb R^3 \rightarrow \Bbb R$, and any real number $b >0$,
and any time $t\geq \alpha |f|_\infty^2/b^2$
\begin{equation*}
| \nabla (e^{t\Delta}f) |_\infty \leq b.
\end{equation*}
\end{Lm}

\begin{proof}
The Lemma is a consequence of the following well known property: there exists $\tilde\alpha >0$ so that
\begin{equation*}
|\partial_i e^{t\partial^2_x}h|_\infty \leq \tilde\alpha t^{-1/2}|h|_\infty, \quad i=1,2,3
\end{equation*}
for any $t \geq 0$ (see for instance Lemma 2.5 in \cite{GG}).
\end{proof}

\begin{Lm} \label{lmduhamel}
There exists a constant $\gamma >0$ so that the following property holds.
Let $t_0, t_+ \in \mathbb R_{\geq 0}$,
$\q(t)\equiv \q(\x,t) \in C^1(\mathbb R^3 \times [t_0, t_0+t_+], \mathbb R^3 )$ be a time dependent vector field, and $c>0$ be so that $|\q|_\infty< c $ for any $t \in [t_0, t_0+t_+]$. For any $t \in [t_0, t_0+t_+]$ introduce the vector field $\w(t)  \equiv \w (\x,t)$ given by
\begin{equation} \label{dew}
 \w(t):=\int_{t_0}^{t} e^{(t-s)\Delta} \q(s)\, ds.
\end{equation}
Then for any  $t \in [t_0, t_0+t_+]$ one has
(i) $ |\w(\cdot, t)|_\infty \leq ct$ and (ii) $ |\nabla \w(\cdot,t) |_\infty \leq \gamma c \sqrt{t-t_0}$. 
\end{Lm}

\begin{rem}
No assumption is made on $\nabla \q$.
\end{rem}

\begin{proof}
Let $t \in [t_0, t_0+t_+]$, $n \in \Bbb Z_{\geq 0}$, and set $h_n:= (t-t_0)/n$. Then the right hand side of
(\ref{dew}) is the limit of the Riemann sum
\begin{equation*}
\w(t)= \lim_{n\rightarrow \infty}\sum_{k=1}^n e^{k h_n \Delta}\q(t_0+(n-k)h_n) h_n
\end{equation*}
and (i) is immediate. Setting $\w:=(w_1,w_2,w_3)$ resp $\q:=(q_1,q_2,q_3)$ one has similarly
\begin{equation} \label{riemann}
\nabla w_i (t)= \lim_{n \rightarrow \infty} \sum_{k=1}^n e^{kh_n \Delta} \nabla q_i(t_0+(n-k)h_n) h_n,
\end{equation}
$i=1,2,3$.
By assumption we have $|f|_\infty \leq c\,h_n$, $f:=h_n q_i(t_0+(n-k)h_n)$. Then by Lemma \ref{lmheat} there exists a constant $\gamma>0$ so that for any time $t_0+s$ with $s >0$ one has 
$|\nabla f|_\infty \leq \frac{\gamma c \,h_n}{\sqrt{s}}$.
Using that the operators $\nabla$ and $e^{t \Delta}$ commute one then has for $s=kh_n$
\begin{equation*}
|e^{kh_n \Delta} \nabla q_i(t_1+(n-k)h_n) h_n|_\infty \leq  \frac{\gamma c \sqrt{h_n}}{\sqrt{k}}.
\end{equation*}
Substituting this estimate in (\ref{riemann}) leads to 
\begin{eqnarray*}
|\nabla w_i (t)|_\infty &\leq&  \lim_{n \rightarrow \infty} \sum_{k=1}^n  \frac{\gamma c \sqrt{h_n}}{\sqrt{k}} 
                 \leq  \gamma c \sqrt{t-t_0} \lim_{n \rightarrow \infty}\frac{1}{\sqrt{n}}\sum_{k=1}^n \frac{1}{\sqrt{k}} \\
                 &\leq & \gamma c \sqrt{t-t_0} \lim_{n \rightarrow \infty}\frac{1}{\sqrt{n}} \int_1^n \frac{2}{\sqrt{x}}dx 
                 \leq  \gamma c \sqrt{t-t_0} \lim_{n \rightarrow \infty}\frac{1}{\sqrt{n}} (\sqrt{n}-1) 
                 \leq  \gamma c \sqrt{t-t_0}.  
\end{eqnarray*}
\end{proof}

\begin{Lm}  \label{lm43} 

There exist constants $\kappa, \tau, \gamma_0, \gamma_1 >0$ so that for any classical solution $(\un,P)$ of (\ref{ns}), any $b>0$, any arbitrarily large $n\in \mathbb Z$ with $b < 2^n$, any $t_0 \geq 0$ so that $|\un(\cdot, t)|_\infty< b$ for any $t_0 \leq t \leq t_0+\tau/4^{n}$, and so that 
$|\nabla \un(\cdot,t_0)|_\infty < 2^n\kappa b$,
one has 
\begin{itemize}
\item[(i)]
$\displaystyle \quad
| \nabla \w(\cdot,t)  |_\infty < 2^{n-2} \kappa b$ for any $t_0\leq t \leq t_0+\tau/4^{n}$,  with\\ $\text{}   \,\,\,\, $ $\displaystyle \w(t):=\int_{t_0}^t e^{(t-s)\Delta}\left[\un. \nabla \un+ \nabla P\right](t_0+s)\, ds $;
\item[(ii)] $\displaystyle \quad
|\nabla \un |_\infty < 2^{n-1} \kappa b $
for any $t_0+\tau/(2 \cdot 4^{n}) \leq t \leq t_0+ \tau/ 4^{n}$;

\item[(iii)] $|\w(\cdot,t)|_\infty \leq  \frac{9}{8}\gamma_0 \kappa 2^n b^2 (t-t_0)$ for any $t_0\leq t \leq t_0+\tau/4^{n}$;

\item[(iv)] $|\nabla \w(\cdot,t)|_\infty \leq  \frac{9}{8} \gamma_1 \kappa 2^n b^2 \sqrt{t-t_0}$ for any $t_0\leq t \leq t_0+\tau/4^{n}$.

\end{itemize}
\end{Lm}

\begin{proof}

Let $\kappa,\tau >0$ be arbitrary, and assume that for some $b >0$, $n \in \mathbb Z$ and $t_0 \geq 0$ one has $|\un(\cdot,t)|_\infty \leq b$ for any $t_0 \leq t \leq t_0+\tau/4^{n}$ and $|\nabla \un(\cdot,t_0)|_\infty < 2^n\kappa b$. To prove Lemma \ref{lm43} we want to apply Lemma \ref{lmheat} to $\tilde\un(t):=e^{(t-t_0)\Delta}\un(t_0)$ and Lemma \ref{lmduhamel} with
\begin{equation*}
\q:= \un.\nabla \un + \nabla P \quad \text{and}\quad \w(t):=\int_{t_0}^t e^{(t-s)\Delta} \q(t_0+s)ds
\end{equation*}
where $P\equiv P(\un)$ denotes the pressure field given by (\ref{PoisP}).
First, let $t_+>0$ be a real number which is small enough so that $|\nabla \un(\cdot,t) |_\infty \leq \frac{9}{8}2^n\kappa b$ for any $t_0 \leq t \leq t_0+t_+$. By 
Theorem \ref{thm2} there exists a constant $\gamma_0 >0$ so that $|(\un. \nabla \un +\nabla P)(\cdot,t)|_\infty \leq \gamma_0 \frac{9}{8} 2^n\kappa b^2$ for any $t_0 \leq t \leq t_0+t_+$.
Lemma \ref{lmduhamel} then says that there is a constant $\gamma_1 >0$ so that for any $t_0\leq t \leq t_0+t_+$
\begin{equation} \label{ineqlm43}
|\w(\cdot,t)|_\infty \leq  \gamma_0\frac{9}{8}  2^n \kappa b^2 (t-t_0) \quad \text{ and } \quad |\nabla \w(\cdot,t)|_\infty \leq  \gamma_1\frac{9}{8} 2^n \kappa b^2 \sqrt{t-t_0}.
\end{equation}
Next we assume $\tau>0$ so that $\frac{9}{8} \gamma_1 \sqrt{\tau} \leq \frac{1}{4}$ and assume $t_+\geq \tau /4^{n} $. Then by the second inequality in (\ref{ineqlm43}) one has $|\nabla w(\cdot,t)|\infty < \frac{\kappa}{4}b^2$ for any $t_0\leq t \leq t_0+\tau/4^{n}$, i.e. with $b \leq 2^n$
\begin{equation} \label{ineq3lm43}
| \nabla \w(\cdot,t)  |_\infty < 2^{n-2} \kappa b
\end{equation}
for any $t_0\leq t \leq t_0+\tau/4^{n}$.
On the other hand, by Lemma \ref{lmheat}, there is a constant $\gamma_2 >0$ so that for any $t> t_0$
\begin{equation} \label{ineq2lm43}
|\nabla \tilde \un(\cdot,t) |_\infty <  \frac{\gamma_2 b}{\sqrt{t-t_0}}.
\end{equation}
Assume  $\kappa>0$ large enough so that 
$\frac{
\gamma_2 b}{ \sqrt{\tau/(2\cdot 4^{n})}} \leq 2^{n-2}\kappa b$, i.e. $ \frac{4\sqrt{2}\gamma_2 }{ \sqrt{\tau}} \leq \kappa $. 
Then Duhamel's principle, (\ref{ineq2lm43}), (\ref{ineq3lm43}), and the fact that $|\nabla \tilde \un(\cdot,t) |_\infty  $ is a decreasing function in time, ensure that
\begin{equation} \label{inviewof}
|\nabla \un(\cdot,t) |_\infty < 2^{n-1}\kappa  b
\end{equation}
for any $t_0+\tau/(2\cdot 4^{n}) \leq t \leq t_0+\tau/ 4^{n}$. It remains to prove that $t_+$ can be taken larger or equal than $\tau/4^{n} $, i.e. to show that for any $t_0 \leq t \leq \tau/4^{n}$,
\begin{equation} \label{toinsureend}
|\nabla \un(\cdot,t) |_\infty \leq \frac{9}{8} 2^n \kappa b.
\end{equation}
To this end it suffices to note that for any $t_0 \leq t \leq t_0+\tau/(4\cdot 4^n)$, by the second inequality in (\ref{ineqlm43}) and $b \leq 2^n$, one has
$|\nabla \w(\cdot,t)|_\infty \leq \frac{1}{8}2^n \kappa b$, 
whereas by the fact that $|\nabla \tilde \un(\cdot,t) |_\infty  $ is a decreasing function in time one has $|\nabla \tilde \un(\cdot,t) |_\infty \leq \kappa 2^n b$. Then Duhamel's principle ensures that (\ref{toinsureend}) holds for any $t_0 \leq t \leq t_0+\tau/(4\cdot 4^n)$. Finally,  for any $t_0+\tau/(4\cdot 4^n)\leq t \leq t_0+\tau/ 4^n$ one has $|\nabla \tilde \un(\cdot,t) |_\infty \leq \frac{1}{2}\kappa b$. Then Duhamel's principle together with (\ref{ineq3lm43}) ensures that (\ref{toinsureend}) also holds for any $t_0+\tau/(4\cdot 4^n)\leq t \leq t_0+\tau/ 4^n$. 
\end{proof}

In the sequel we assume that all solutions $(\un, P)$ of (\ref{ns}) we consider have initial data $\un^0$ in $ L^3(\mathbb R^3) \cap L^\infty(\mathbb R^3)$. Recall the solution $\un \in L^{3,\infty}(\mathbb R^3 \times [0,t_b])$ generated by $\un^0$ is smooth, where
$t_b$ is its blow-up time.

\begin{Lm}  \label{cadregeneral}
There exists a constant $ \theta >0$ so that for any solution $\un\equiv \un(\x,t) $ of (\ref{ns}) as above, for any $t_0 \in [0,t_b)$  and for any $n \in \mathbb Z$ with $|\un(\cdot,t_0)|_\infty < \theta 2^{n}$, there exists $t_1,t_2$ with 
$t_0<t_1 \leq t_0+\frac{\tau}{6} 4^{-n+1}$ and $t_1+\frac{\tau}{2} 4^{-n+1}\leq t_2 < t_b$, where $\tau$ is the constant in Lemma \ref{lm43} so that
\begin{equation*}
(i)\, |\un(\cdot,t)|_\infty \leq \frac{9}{8}|\un(\cdot,t_0)|_\infty \,\,\, \forall \,t_0\leq t \leq t_2 \quad \text{and} \quad (ii)\,|\nabla \un(\cdot, t)|_\infty < \eta 4^n \,\,\, \forall\,t_1 \leq t\leq t_2,
\end{equation*}
with $\eta:=\theta \kappa$, where $\kappa>0$ is the constant in  Lemma \ref{lm43}.
\end{Lm}

\begin{proof}
Let $t_0 \in [0,t_b)$ and set $t_+:=\sup\{t>0\,|\,|\un(\cdot,t_0+t)|_\infty \leq 9/8 |\un(\cdot,t_0)|_\infty\}$.
It is well known that $\un$ is smooth as long as $|\un(\cdot,t)|_\infty$ remains finite (see for instance \cite{ESS}).
Thus one has $t_0+t_+ < t_b$, where $t_b$ is the blow up time of $\un$. 
Let $n \in \mathbb Z$ be an arbitrary large integer with $9/8|\un(\cdot,t_0)|_\infty <2^n$. Let $t_\ast^{(n)} \in [t_0,t_0+t_+]$, with $t_\ast + \tau/4^n \leq t_0+t_+$, be so that
\begin{equation} \label{179i}
  M_{n-1}=1/2M_n \leq |\nabla \un (\cdot,t_\ast^{(n)})|_\infty \leq M_n,  
\end{equation} 
where for any $n \in \mathbb Z$ we set $M_n:=\kappa 2^{n}|\un(\cdot,t_0)|_\infty $, and $\kappa,\tau$ are the constants in Lemma \ref{lm43}. Lemma \ref{lm43}, applied at $t_\ast^{(n)}$, $n$ and $b=9/8|\un(\cdot,t_0)|_\infty $, then ensures the existence of $t_\ast^{(n-1)} \in [t_0,t_0+t_+]$, with
$\tau/(2\cdot 4^n) >t_\ast^{(n-1)}-t_\ast^{(n)}>0$, and 
\begin{equation} \label{179ii}
 M_{n-2} \leq |\nabla \un (\cdot,t_\ast^{(n-1)})|_\infty \leq M_{n-1},
\end{equation}
and so that
\begin{equation} \label{179iib}
|\nabla \un (\cdot,t)|_\infty \leq M_{n-1} \text{ for any } t_{\ast}^{(n-1)}\leq t \leq t_{\ast}^{(n)}+1/(\rho 4^n)<t_0+t_+.
\end{equation}
Moreover, Duhamel's principle and item (iii) in Lemma \ref{lm43} ensure that
there is a constant $\gamma_3>0$ so that for any $t \in [t_\ast^{(n)},t_\ast^{(n-1)}]$
\begin{equation*}
|\un(\cdot,t)-\un(\cdot,t_\ast^{(n)})|_\infty \leq \frac{81}{64}\gamma_0 \kappa 2^n |\un(\cdot,t_0)|_\infty^2 (t-t_\ast^{(n)})\leq \frac{\gamma_3}{2^n}|\un(\cdot,t_0)|_\infty^2,
\end{equation*}
where $\gamma_0$ is the constant in Lemma \ref{lm43}.
We have thus established the following property: let $\check{n} \in \mathbb Z$ be so that $ M_{\check{n}-1} \leq |\nabla \un (\cdot,t_0)|_\infty \leq M_{\check{n}}$ and $2^{\check{n}} >9/8|\un(\cdot,t_0)|_\infty$. Then for any $n_0 \in \mathbb Z$ sufficiently large so that
$\gamma_3 |\un(\cdot,t_0)|_\infty^2 \sum_{n \geq n_0}1/2^n \leq 1/8 |\un(\cdot,t_0)|_\infty
$ and $2^{n_0} >9/8|\un(\cdot,t_0)|_\infty$, i.e.
\begin{equation} \label{sufflarge}
\frac{9}{8}|\un(\cdot, t_0)|_\infty \leq \text{max}\left(\frac{1}{16 \gamma_3}2^{n_0}, 1\right),
\end{equation}
one has 
\begin{equation*}
t_\ast^{(n_0)}:=\sum_{\check{n} > n\geq n_0} (t_\ast^{(n)}- t_\ast^{(n+1)}) \in \left[t_0, t_0+\frac{1}{2}\sum_{n \geq n_0}\tau/4^{n+1}\right],
\end{equation*}
i.e., $t_\ast^{(n_0)} \in [t_0,t_0+\frac{\tau}{6}4^{-n_0}] $, and for any  $\displaystyle t_0 \leq t \leq t_\ast^{(n_0)}$
\begin{equation} \label{1792011end}
|\un(\cdot,t)|_\infty \leq |\un(\cdot,t_0)|_\infty+  \gamma_3 |\un(\cdot,t_0)|_\infty^2 \sum_{n \geq n_0}1/2^n \leq \frac{9}{8}|\un(\cdot,t_0)|_\infty, 
\end{equation}
i.e., $[t_0,t_\ast^{(n_0)}]\subseteq[t_0,t_0+t_+]$.
Moreover, for $n_0$ as above on has by (\ref{179i}) and (\ref{179ii})
\begin{equation*} 
 M_{n_0-1} \leq |\nabla \un (\cdot,t_\ast^{(n_0)})|_\infty \leq M_{n_0}=\kappa 2^{n_0}|\un(\cdot,t_0)|_\infty,
\end{equation*}
and by (\ref{179iib})  for any $ t_{\ast}^{(n_0)}\leq t \leq t^{(n_0)}_{\ast\ast}:=t_{\ast}^{(n_0+1)}+\tau/4^{n_0+1} <t_0+t_+$
\begin{equation} \label{1792011l}
|\nabla \un (\cdot,t))|_\infty \leq M_{n_0} \quad (=\kappa 2^{n_0}|\un(\cdot,t_0)|_\infty).
\end{equation}
Finally note that $0 < t_{\ast}^{(n_0)} \leq \frac{\tau}{6}4^{-n_0}$ and $0< t^{(n_0-1)}-t^{(n_0)}<\tau/(2 \cdot 4^{n_0})$. Then  $t_\ast^{(n_0)} +\tau/(2 \cdot 4^{n_0+1}) \leq t^{(n_0)}_{\ast\ast}\leq t_+$.
Let $ n \in \mathbb Z$ be so that $|\un(\cdot,t_0)|_\infty \leq \theta 2^n$, with $\theta:=\frac{8}{9}\text{max}\left(\frac{1}{16 \gamma_3},1  \right)$. Then we have proved that
for any $t_0 \leq t \leq t_2:=t_{\ast\ast}^{n_0}$ one has $|\un(\cdot,t)|_\infty <2^n$, and for any $t_1 \leq t \leq t_2$ with $t_1:=t^{(n_0)}_\ast$, $|\nabla \un (\cdot,t)|_\infty <\eta 4^n$, with $\eta:=\kappa \theta$.
\end{proof}

\begin{proof}[Proof of Proposition \ref{thm3}]
Let $\tau >0$ be the constant in Lemma \ref{lm43}.
Let $ \eta$ and $\theta$ be the constants in Lemma \ref{cadregeneral} and
let $\un$ be the solution of (\ref{ns}) generated by an initial datum $\un^0$ as in the statement of Proposition \ref{thm3}. Recall that $\un(\cdot,t)$ remains smooth in $\x$ as long as $|\un(\cdot,t)|_\infty $ remains finite. Let $n_0 \in \mathbb Z_{\geq 0}$ be so that $|\un^0|_\infty < \theta 2^{n_0}$. Then by Lemma \ref{cadregeneral}, there exist $T_+^{(n_0)}(0),T_{++}^{(n_0)}(0)$ with $0 < T_+^{(n_0)}(0) \leq \frac{\tau}{6}4^{-n_0+1}$ and $ T_+^{(n_0)}(0) + \frac{\tau}{2}4^{-n_0+1}\leq T_{++}^{(n_0)}(0)$, so that one has $|\un(\cdot,t)|_\infty \leq \frac{9}{8}\theta 2^{n_0}  $ for any $0\leq t \leq T_{++}^{(n_0)}(0)$, whereas for any $t \in [T_+^{(n_0)}(0),T_{++}^{(n_0)}(0)]$ one has $|\nabla \un(\cdot, t)|_\infty < \eta 4^{n_0} $, i.e.
\begin{equation} \label{septbound}
|\nabla \un(\cdot, t)|_\infty \leq \beta \text{max}_{0 \leq s \leq t}|\un(\cdot, s)|_\infty^2 \quad \text{with} \quad \beta:= 4 \eta/\theta^2.
\end{equation}
Moreover, arguing as above we have the following property: for any $t_\ast >T_+^{(n_0)}(0)$ with $\text{max}_{0 \leq s \leq t_\ast}|\un(\cdot, s)|_\infty < \theta 2^{n_0}$ and $|\nabla \un(\cdot, t_\ast)|_\infty < \eta 4^{n_0}$ one 
has (\ref{septbound}) for any $t \in [t_\ast+T_+^{(n_0)}(0)(t_\ast),T_{++}^{(n_0)}(0)]$,  and again one has
$0 < T_+^{(n_0)}(t_\ast) \leq \frac{\tau}{6}4^{-n_0+1}$ and $ T_+^{(n_0)}(t_\ast) + \frac{\tau}{2}4^{-n_0+1}\leq T_{++}^{(n_0)}(t_\ast)$.
Moreover one has
\begin{equation*}
|\nabla \un(\cdot, t)|_\infty < 5\eta 4^{n_0}
\end{equation*}
for any $t \in [t_\ast, t_\ast+ T_+^{(n_0)}(t_\ast)]$.
To obtain the latter bound we use $\eta=\theta \kappa$ where $\kappa$ is the constant in Lemma \ref{lm43},
to conclude $|\nabla \un(\cdot,t_\ast)|_\infty < 2^{n_0+2} \kappa |\un(\cdot,t_\ast)|_\infty$, and then apply Duhamel's principle together with Lemma \ref{lm43} (i) with $n=n_0+2$ to obtain $|\nabla \un)\cdot,t)|_\infty < (2^{n_0+2} + 2^{n_0})\kappa |\un(\cdot,t_\ast)|_\infty $ for any
$t_\ast \leq t \leq \tau 4^{-n_0-2}$.

Then 
\begin{equation} \label{septbound2}
|\nabla \un(\cdot, t)|_\infty \leq \beta_3 \text{max}_{0 \leq s \leq t}|\un(\cdot, s)|_\infty^2 \quad \text{with} \quad \beta_3:= 5 \beta=20 \eta/\theta^2.
\end{equation}
for any $t \in [t_\ast, t_\ast+T_+^{(n_0)}(t_\ast)]$, i.e. we have established the conclusion of Theorem \ref{thm3} for any 
$t^{(n_0)}_+ \leq t < T^{(n_0)}:=\sup \{t \in \mathbb R\,|\, \text{max}_{0 \leq s \leq t_\ast}|\un(\cdot, s)|_\infty^2 < \theta 2^{n_0}  \}$.
Then, if $T^{(n_0)}=+\infty$, the proof is complete. Otherwise argue as above to prove (\ref{septbound2}) (with the same $\beta_3$) for any $t \in [T^{(n_0)},T^{(n_0+1)}]$, $t \in [T^{(n_0+1)},T^{(n_0+2)}]$,.... If $\un$ does not blow up, then it suffice to iterate the argument finitely many steps, otherwise an infinite number of steps. 
\end{proof}

\begin{proof}[Proof of Corollary \ref{intro}]
Assume that $\un$ is a solution of (\ref{ns}) generated by an initial datum $\un^0$  as in Proposition \ref{thm3}.
Let $t_\ast  \in [t_0,t_b)$, where $t_b\equiv t_b (\un^0) \in \mathbb R_{>0}\cup \{+\infty\}$ is the blow up time, and $t_0 \equiv t_0(\un^0)$ is the time given by Proposition \ref{thm3}. By Proposition \ref{thm3} one has $|\un(\cdot,t)|_\infty \leq \frac{9}{8}|\un^0|_\infty$ for any
$0 \leq t \leq t_0$. Assume $t_b< \infty$ and let $M > \frac{9}{8}|\un^0|_\infty$ be arbitrarily fixed. Recall that $|\un(\cdot,t)|_\infty \rightarrow +\infty$ when $t \rightarrow t_b$. Thus 
there exists $t_\ast:=\text{min}\{t>0\,|\, |\un(\cdot, t)|_\infty =2M\} \in [t_0, t_b)$, and by Proposition \ref{thm3}
\begin{equation} \label{distance}
|\nabla \un(\cdot, t_\ast|_ \infty \leq \beta_2|\un(\cdot,t_\ast)|\infty^2=4\beta_2M^2.
\end{equation}
Let $\p \in \mathbb R^3$ be so that $|\un(\p, t_\ast)|=2M$ and consider 
\begin{equation*}
r:= \text{dist}\left(\p,\{\x \in \mathbb R^3\,|\,|\un(\x,t_\ast|< M  \}\right).
\end{equation*}
Then by (\ref{distance}) there exists a constant $\tilde \beta>0$ so that $r > \tilde \beta/M$, and the kinetic energy
$E$ supported on $\{\x \in \mathbb R^3\,|\,|\un(\x,t_\ast| \geq M \}$ satisfies
$ \displaystyle E \geq 4\pi \tilde \beta^3/(3M)$.
By the energy inequality one has $E \leq |\un^0|_2^2$. Then necessarily $|\un^0|_2^2 \geq \frac{4\pi \tilde \beta^3}{3M}$, i.e., 
\begin{equation} \label{endcontradiction}
M > \frac{4\pi \tilde \beta^3}{3|\un^0|_2^2 }. 
\end{equation}
Assume that the condition 
$\displaystyle |\un|_\infty \leq 4\pi \tilde \beta^3/(3|\un^0|_2^2)$
is satisfied and fix $M \in \{\frac{9}{8}|\un^0|_\infty,\frac{4\pi \tilde \beta^3}{3|\un^0|_2^2 }\}$. Then (\ref{endcontradiction}) shows that the time $t_\ast:=\text{min}\{t>0\,|\, |\un(\cdot, t)|_\infty =2M\}>t_0$ cannot exist, i.e., $|\un(\cdot, t)|_\infty < 2M$ for any $t \geq 0$, thus $\un$ remains smooth for all time. 
\end{proof}

\noindent Philipp Lohrmann\\
Dipartimento di Matematica e Applicazioni  \textquotedblleft R. Cacciopoli" Università Federico II, Napoli\\
Via Cintia, Monte S. Angelo,
I-80126 Napoli, Italy \\
E-mail address: \underline{philipp.lohrmann@unina.it}

\end{document}